\documentclass[11pt,english]{article}

\usepackage{amsmath,amssymb,amsthm}
\usepackage{multicol}
\usepackage[table]{xcolor}
\usepackage[margin=1in]{geometry}
\usepackage{graphicx,color}
\usepackage{enumitem}
\usepackage{fullpage}
\usepackage[noblocks]{authblk}
\usepackage{tcolorbox}
\usepackage{babel}
\usepackage{wrapfig}
\usepackage{MnSymbol}
\usepackage{mdframed}
\usepackage{mathtools}
\usepackage{floatrow}
\usepackage{multirow}
\usepackage{booktabs}
\usepackage{palatino}

\usepackage[ruled,linesnumbered,vlined]{algorithm2e}
\usepackage{tablefootnote}
\usepackage{thm-restate}
\usepackage{bbding}
\usepackage{pifont}
\usepackage{nicefrac}

\newcommand{\R}{\mathbb{R}}

\newcommand{\down}[1]{\left\lfloor #1\right\rfloor}
\newcommand{\up}[1]{\left\lceil #1\right\rceil}

\newcounter{magicrownumbers}

\graphicspath{{./figures/}}

\usepackage{tablefootnote}

\definecolor{Darkblue}{rgb}{0,0,.8}
\definecolor{Brown}{cmyk}{0,0.61,1.,0.60}
\definecolor{Purple}{cmyk}{0.45,0.86,0,0}
\definecolor{Darkgreen}{rgb}{0.133,0.500,0.133}

\definecolor{MyGreen}{rgb}{0.200,0.500,0.200}
\renewcommand{\emph}[1]{{\color{MyGreen}{\em #1}}}

\usepackage[colorlinks,linkcolor=Darkblue,filecolor=blue,citecolor=blue,urlcolor=Darkblue,pagebackref]{hyperref}
\usepackage[nameinlink]{cleveref}

\usepackage[colorinlistoftodos,prependcaption,textsize=tiny]{todonotes}

\newcommand{\namedref}[2]{\hyperref[#2]{#1~\ref*{#2}}}
\newcommand{\propref}[1]{\hyperref[#1]{property~(\ref*{#1})}}

\newcommand{\IGNORE}[1]{}

\newtheorem*{theorem*}{Theorem}
\newtheorem{theorem}{Theorem}[section]
\newtheorem{lemma}[theorem]{Lemma}

\newtheorem{observation}[theorem]{Observation}
\newtheorem{corollary}[theorem]{Corollary}

\newtheorem*{question*}{Question}

\newtheorem{remark}[theorem]{Remark}

\newtheorem*{conjecture*}{Conjecture}

\newcommand{\old}[1]{{}}

\hypersetup{
  colorlinks=true,
  linkcolor=red!70!black,
  linktoc=all,
  citecolor=violet,
    pdfpagemode=FullScreen,
}

\title{Sparse Bounded Hop-Spanners for\\ Geometric Intersection Graphs}

\date{}

 \author{Sujoy Bhore\thanks{Department of Computer Science \& Engineering, Indian Institute of Technology Bombay, Mumbai, India.\\ Email: \href{sujoy@cse.iitb.ac.in}{sujoy@cse.iitb.ac.in}}
 \quad
 Timothy M. Chan\thanks{Siebel School of Computing and Data Science, University of Illinois at Urbana-Champaign.\\ Email: \href{tmc@illinois.edu}{tmc@illinois.edu}. 
 Work supported in part by
NSF Grant CCF-2224271.}
\quad
Zhengcheng Huang\thanks{Siebel School of Computing and Data Science, University of Illinois at Urbana-Champaign. Email: \href{zh3@illinois.edu}{zh3@illinois.edu}}\\
\quad
Shakhar Smorodinsky\thanks{Department of Computer Science, Ben-Gurion University of the Negev, Be'er Sheva, Israel.\\ Email: \href{shakhar@bgu.ac.il}{shakhar@bgu.ac.il}. Partially supported by Grant 1065/20 from the Israel Science Foundation, by the United States – Israel Binational Science Foundation (NSF-BSF grant no. 2022792) and by ERC grant no.~882971 "GeoScape" and by the Erd{\H o}s Center.}
\quad
Csaba D. T\'{o}th\thanks{Department of Mathematics, California State University Northridge, Los Angeles, USA and Department of Computer Science, Tufts University, Medford, MA, USA Email: \href{csaba.toth@csun.edu}{csaba.toth@csun.edu}. Research supported in part by the NSF award DMS-2154347}}

\begin{document}
\maketitle

\begin{abstract}
    We present 
    new results on $2$- and $3$-hop spanners for geometric intersection graphs. These include improved upper and lower bounds for $2$- and $3$-hop spanners for many geometric intersection graphs in $\R^d$.
    For example, we show that the intersection graph of $n$ balls in $\R^d$ admits a $2$-hop spanner of size $O^*\left(n^{\frac{3}{2}-\frac{1}{2(2\lfloor d/2\rfloor +1)}}\right)$ and the intersection graph of $n$ fat axis-parallel boxes in $\R^d$ admits a $2$-hop spanner of size $O(n \log^{d+1}n)$. 
    
    Furthermore, we show that the intersection graph of general semi-algebraic objects in $\R^d$  admits a $3$-hop spanner of size $O^*\left(n^{\frac{3}{2}-\frac{1}{2(2D-1)}}\right)$, where $D$ is a parameter associated with the description complexity of the objects. For such families (or more specifically, for tetrahedra in $\R^3$), we provide a lower bound of $\Omega(n^{\frac{4}{3}})$.
    For $3$-hop and axis-parallel boxes in $\R^d$, we provide the upper bound $O(n \log ^{d-1}n)$ and lower bound $\Omega\left(n (\frac{\log n}{\log \log n})^{d-2}\right)$.
\end{abstract}

\section{Introduction}
A \textit{spanner} of a graph $G$ is a subgraph $G'$ that approximately preserves the distances between pairs of vertices in $G$. More formally, let $G=(V,E)$ be a simple graph. A spanning subgraph $G'=(V,E')$ is called a \emph{$t$-spanner} for $G$ if for every two vertices $x,y\in V$ $d_{G'}(x,y) \leq t \cdot d_G(x,y)$ where $d_G(x,y)$ denotes the shortest path distance between $x$ and $y$ in $G$. Here, \( t \geq 1 \) is a constant called the \textit{stretch factor} (or \textit{stretch} for short). Spanners are an important concept in graph theory and computer science, particularly in the context of geometric intersection graphs, where vertices represent geometric objects (such as approximate location, e.g., disks, rectangles) and edges represent intersections between these objects. The ability to construct spanners for geometric intersection graphs with fewer edges is crucial because it enables efficient algorithms for problems such as shortest path computation, network design, and routing while maintaining proximity to the original geometric distances. By using spanners, one can often reduce the complexity of geometric problems without sacrificing the quality of the solution (see, e.g., \cite{GudmundssonNS16, Li03, mitchell2017proximity, narasimhan2007geometric, Rajaraman02}).
\\

The following theorem is a cornerstone in the theory of spanners for general graphs:
\begin{theorem}[\cite{althofer1993sparse}]\label{thm:althofer}
    Let $G=(V,E)$ be a simple undirected graph on $n$ vertices. Let $t \geq 1$ be a fixed integer parameter. Then $G$ contains a $(2t-1)$-spanner $G'$ with $O(n^{1+\frac{1}{t}})$ edges. 
\end{theorem}
Note that the above theorem is non-trivial only for $t \geq 2$ or stretch factor at least $3$. Indeed, one cannot hope to obtain a $2$-spanner with $o(n^2)$ edges in general as is exemplified by the complete bi-partite graph $K_{\lfloor \frac{n}{2}\rfloor,\lceil \frac{n}{2}\rceil}$.
The general bound $O(n^{1+\frac{1}{t}})$ is conjectured to be optimal due to the \emph{Erd\H{o}s Girth Conjecture}~\cite{Erdos64extremalproblems}, and the spanner $G'$ of this size can be computed by a greedy algorithm. It would be interesting to ask what families of graphs admit 2-spanners with $o(n^2)$ edges or $3$-spanners with $o(n^{3/2})$ edges.

The \textbf{intersection graph} of a set system $(X,\mathcal{F})$ is a graph $G=(\mathcal{F},E)$, where the vertices correspond to the sets in $\mathcal{F}$, and there is an edge between sets $A,B\in \mathcal{F}$ if and only if $A\cap B\neq \emptyset$. Every edge has unit weight. For geometric intersection graphs, $t$-spanners are also called \emph{$t$-hop spanners} to emphasize the use of hop distance.
In this paper, our goal is to obtain hop spanners for geometric intersection graphs that are smaller than guaranteed by \Cref{thm:althofer}.

\noindent \paragraph{Previous work on spanners in intersection graphs.}
Initial work on spanners for geometric intersection graphs was motivated by wireless communication, and was limited to unit disk graphs in the plane. In this setting, spanners for weighted graphs were also considered, where the weight of an edge is the distance between the disk centers~\cite{FurerK12,GaoGHZZ05,GaoZ05}. Recent work focused on hop spanners for unit disk graphs~\cite{biniaz2020plane,CatusseCV10,DumitrescuGT22}, and intersection graphs of other geometric objects (including string graphs)~\cite{ChanH23,ConroyT23}.
Most previous bounds are constrained to objects in the plane. 
Chan and Huang~\cite{ChanH23} use shallow cuttings to show that any family of $n$ objects, with union complexity bounded by $\mathcal{U}(n)$, admits a $2$-hop spanner of size $O(\mathcal{U}(n) \log n)$. This yields an $O(n\log n)$ bound for disks or pseudodisks in the plane~\cite{KedemLPS86}; 
$n\log n\cdot 2^{O(\log^* n)}$ for fat objects of constant algebraic description complexity~\cite{AronovBES14}, and in particular $O(n\log n \log^*n )$ for fat triangles in the plane~\cite{AronovBES14}. 
Lower bounds can be derived from constructions for geometric intersection graphs with girth $2t+1$, as a $t$-spanner contains the entire graph~\cite{basit21,Davies21,tomonzakharov2021}. In particular, Davies~\cite{Davies21} constructed families of axis-aligned boxes in $\R^3$ whose intersection graph has unbounded girth and  $\omega(n)$ edges. The boxes in Davies' construction project to squares in the $xy$-plane, so they can be lifted to hypercubes in $\R^4$ with the same intersection pattern.
See \Cref{tab:old-results} for a list of known results.

\begin{table}\centering
\begin{tabular}{llllll}
Hops &  Objects & Dim. & Upper Bounds & Lower Bounds &Ref.\\\hline
2  & unit disks  & 2 & $O(n)$ & & \cite{ConroyT23}\\
   & translates of a convex body & 2 &  $O(n)$ &  & \cite{ConroyT23}\\
   & axis-aligned fat rectangles & 2 & $O(n \log n)$ & $\Omega(n\frac{\log n}{\log \log n})$ & \cite{ConroyT23}\\ 
   & disks of arbitrary radii ($\ast$)  & 2 & $O(n \log n)$ & $\Omega(n\frac{\log n}{\log \log n})$ & \cite{ChanH23,ConroyT23}\\ 
   & general fat objects ($\ast\ast$) & 2 & $n\log n\cdot 2^{O(\log^* n)}$ & $\Omega(n\frac{\log n}{\log \log n})$ & \cite{ChanH23,ConroyT23}\\ 
\hline
3   & general fat objects  & $d\ge 2$ & $O(n \log n)$ & &\cite{ChanH23}\\ 
   & axis-aligned rectangles & 2 & $O(n \log n)$ & & \cite{ChanH23}\\ 
  & strings & 2 & $O(n \log^3 n)$  & & \cite{ChanH23}\\ 
\hline 
$O_k(1)$ & strings & 2 & $O(n \alpha_k(n))$ & & \cite{ChanH23}\\ 
$O(1)$ & axis-aligned boxes & $d\geq 3$ & & $\omega(n)$ & \cite{Davies21}\\
$O(1)$ & axis-aligned hypercubes & $d\geq 4$ & $O(n \alpha_k(n))$ & $\omega(n)$ & \cite{ChanH23,Davies21}\\
$O_{d,k}(1)$ & general fat objects & $d$ &  $O(n \alpha_k(n))$ & & \cite{ChanH23}
\\\hline                                                 
\end{tabular}
\caption{Summary of previously known results. The trivial lower bound in each case is $\Omega(n)$. 
The upper bound in ($\ast$) generalizes to $O(\mathcal{U}(n) \log n)$ for objects with union complexity $\mathcal{U}(.)$; and the lower bound holds for homothets of any convex body in the plane. 
The upper bound in ($\ast\ast$) holds for fat objects 
of constant algebraic description complexity.
The function $\alpha_k(.)$ denotes the $k$-th function in the inverse Ackermann hierarchy: $\alpha_0(n)=n/2$, $\alpha_1(n)=\log n$, $\alpha_2(n)=\log^* n$ (the iterated logarithm), $\alpha_3(n)=\log^{**} n$ (the iterated iterated logarithm), 
etc.}
\label{tab:old-results}
\end{table}

\subsection{New results}
In this paper, we focus on 2- and 3-spanners. Recall that in general there is no
hope to obtain a $2$-spanner of sub-quadratic size (in $n$) nor a $3$-spanner of $o(n^{3/2})$ size for general graphs with $n$ vertices. No sub-quadratic bounds were previously known for $2$-spanners for intersection graphs in dimension $3$ or more even for unit balls. 
And no sub-$n^{3/2}$ bounds were previously known for 3-spanners for intersection graphs of non-fat objects in dimension above 2, even for simplices.
We give the first 2-spanner constructions of sub-quadratic size (in fact, $O(n^{3/2})$ or better)
for many types of fat objects, including balls, and we give the first 3-spanner constructions of sub-$n^{3/2}$ size
for many types of non-fat objects, including simplices.
Furthermore, we complement our upper bounds with a number of lower bounds ruling out near-linear-size spanners for many types of objects.
See \Cref{tab:new-results} for a detailed list of all results.\footnote{
In this paper, we use the $O^*$ notation to hide $n^\varepsilon$ factors for an arbitrarily small constant $\varepsilon>0$.
}
Our results answer open questions by Chan and Huang~\cite{ChanH23}.\footnote{
At the very end of their paper,
Chan and Huang~\cite{ChanH23} specifically wrote: ``Another question is whether near-linear bounds are possible for other intersection graphs
 not addressed here, e.g., for simplices in $\R^3$. Here, one might want to start more modestly
 with any upper bound better than for general graphs.''  Our $\Omega(n^{4/3})$ lower bound answers the question in the first sentence in the negative for simplices in $\R^3$,
while our upper bounds show that the goal in the second sentence is attainable at least in the case of 3 hops.
}

One of our main results deals with 2- and 3-spanners for semi-algebraic graphs. The notion of semi-algebraic graphs was the focus of several recent papers. Most notably, Fox et al.~\cite{FPSSZ17} provided improved bounds on Zarankiewicz's problem for the class of semi-algebraic graphs. Informally, a semi-algebraic graph $G=(V, E)$ is a bipartite graph where the vertices can be represented as points in $\R^d$ for some fixed small constant dimension $d$ and for every pair of vertices $u,v$ we have that $\{u,v\} \in E$ if some boolean combination of a constant number of inequalities involving $2d$-variate polynomials of constant degree is satisfied, where the variables in the polynomials are the coordinates of the two points representing $u$ and $v$ (see below for a formal definition). Almost all intersection graphs in the literature are, in fact, semi-algebraic. One of our main results is a bound of the form $O(n^{\frac{3}{2} -\delta_d})$ for 3-hop spanner for arbitrary intersection graphs of semi-algebraic sets in $\R^d$ and for 2-hop for semi-algebraic sets that are also ``fat'' where $\delta_d$ is a constant that depends on the description complexity of the objects.  In this regard, it is important to note that without the fatness assumption, no sub-quadratic 2-hop spanner can be obtained as the graph $K_{n,n}$, which has $n^2$ edges, can be realized as the intersection of $n$ horizontal and $n$ vertical narrow rectangles, and any 2-hop spanner must contain all $n^2$ edges. Also, a common property for
semi-algebraic intersection graphs is that they all have bounded VC-dimension. However, bounded VC-dimension by itself does not guarantee $o(n^{\frac{3}{2}})$ edges as demonstrated later in \Cref{thm:bipartite-boundedVC-LB}. Thus, the extra geometric properties that these graphs possess are crucial in order to obtain such bounds. Perhaps not surprisingly, one of our main tools for tackling this problem is an application of the divide-and-conquer, specifically the \textbf{polynomial partitioning} technique. This technique was introduced by Guth and Katz in their groundbreaking work on the \emph{Erd\H{o}s distinct distances} problem~\cite{GuthKatz2015}. This method involves partitioning the space into cells defined by low-degree polynomials. A key result states that for a set of $n$ points in $\R^d$, and for any $D>0$, there exists a $d$-variate polynomial $f$ of degree $D$ such that the zero set $Z(f)$ of $f$ partitions $\R^d$
into $O(D^d)$ cells (i.e., connected components of $\mathbb{R}^d \setminus Z(f)$), each containing in its interior at most $\frac{n}{D^d}$ points.

\begin{table}\centering
\begin{tabular}{lllll}
Hops &  Objects & Dimension & Upper Bounds & Lower Bounds\\\hline
2 &  general fat objects & $d\ge 3$ & $O(n^{3/2})$ & $\Omega(n^{4/3})$\\
 &  general fat objects ($\ast$) & $d\ge 3$ & $O^*(n^{\frac{3}{2}-\frac{1}{2(2D-1)}})$ & $\Omega(n^{4/3})$\\
  & fat simplices ($\ast\ast$) & 3 & $O^*(n^{10/7})$  & $\Omega(n^{4/3})$ \\
  &                & $d\ge 4$ & $O(n^{\frac{3}{2}-\frac{1}{O(d^2)}})$ & $\Omega(n^{4/3})$ \\
  &  balls ($\dag$)  & 3 & $O^*(n^{4/3})$ & \\
  &                & 4 & $O^*(n^{7/5})$ & \\
  &                & 5 & $O^*(n^{7/5})$ & $\Omega(n^{4/3})$ \\
  &                & $d\ge 6$ & $O^*(n^{\frac{3}{2}-\frac{1}{2(2\down{d/2}+1)}})$ & $\Omega(n^{4/3})$\\
  & fat axis-aligned boxes ($\ddag$)       & $d\ge 3$ & $O(n\log^{d+1}n)$ & $\Omega(n(\frac{\log n}{\log\log n})^{\down{d/2}-1}$)\\\hline
3 & general objects ($\ast$) & $d\ge 3$ & $O^*(n^{\frac{3}{2}-\frac{1}{2(2D-1)}})$ & $\Omega(n^{4/3})$\\
  & simplices       & 3  & $O^*(n^{10/7})$ & $\Omega(n^{4/3})$ \\
  &                               & $d\ge 4$ & $O(n^{\frac{3}{2}-\frac{1}{O(d^2)}})$ & $\Omega(n^{4/3})$ \\
  & axis-aligned boxes            & $d\ge 3$ & $O(n\log^{d-1}n)$ & $\Omega(n(\frac{\log n}{\log\log n})^{d-2}$)\\\hline                                                 
\end{tabular}
\caption{Summary of new results.  Objects for ($\ast$) are assumed to be semi-algebraic for some parameter $D$ associated
with their description complexity.
The upper bounds for ($\ast\ast$) hold more generally for fat polyhedra of constant complexity, while
the lower bounds for ($\ast\ast$) hold more specifically for unit regular tetrahedra or unit non-axis-aligned cubes in $\R^3$.
The lower bounds for ($\dag$) hold more specifically for unit balls.
The lower bound for ($\ddag$) holds specifically for unit axis-aligned hypercubes.
}

\label{tab:new-results}
\end{table}

\old{
\begin{table}[htp]
\begin{center}
\begin{tabular}{ |c|c|c|c| }
\hline
Space & Geometric Objects &  Stretch & Spanner Size  \\
\hline
$\mathbb{R}^2$  & Unit disks  & 2 & $O(n)$~\cite{ConroyT22}\\
\hline
$\mathbb{R}^2$  &  Translates of a convex body & 2 &  $O(n)$~\cite{ConroyT22}\\
\hline
$\mathbb{R}^2$ & Axis-aligned fat rectangles & 2 & $\Omega(n\log n/\log \log n)$, $O(n \log n)$ \cite{ConroyT22}\\ 
\hline
$\mathbb{R}^2$ & Objects with union complexity $\mathcal{U}(.)$ & 2 & $O(\mathcal{U}(n) \log n)$ \cite{ChanH23}\\ 

\hline
$\mathbb{R}^2$ & Fat convex bodies & 3 & $O(n \log n)$ \cite{ConroyT22}\\ 
\hline
$\mathbb{R}^2$ & Axis-aligned rectangles & 3 & $O(n \log n)$ \cite{ChanH23}\\ 
\hline
$\mathbb{R}^2$ & Strings & 3 & $O(n \log^3 n)$ \cite{ChanH23}\\ 
\hline
$\mathbb{R}^2$ & Strings & $O_k(1)$ & $O(n \alpha_k(n))$ \cite{ChanH23}\\ 
\hline
$\mathbb{R}^d$ & Fat objects & 3 & $O(n \log n)$ \cite{ChanH23}\\ 
\hline
$\mathbb{R}^d$ & Fat objects & $O_{d,k}(1)$ & $O(n \alpha_k(n))$ \cite{ChanH23}\\ 
\hline
\end{tabular}
\vspace{.2cm}
\caption{\label{tab:2} Current best bounds for $t$-spanners for geometric intersection graphs.  The function $\alpha_k(.)$ denotes the $k$-th function in the inverse Ackermann hierarchy: $\alpha_0(n)=n/2$, $\alpha_1(n)=\log n$, $\alpha_2(n)=\log^* n$ (the iterated logarithm), $\alpha_3(n)=\log^{**} n$ (the iterated iterated logarithm), 
etc.}
\end{center}
\end{table}
}
\old{
\begin{table}\centering
\begin{tabular}{ |c|c|c|c|c|}
\hline
Strech &  Objects & Dimension & Upper Bounds & Lower Bounds\\
\hline
2 &  general fat objects ($\ast$) & $d\ge 3$ & $O^*(n^{\frac{3}{2}-\frac{1}{2(2D-1)}})$ & $\Omega(n^{4/3})$\\
\hline
  2 & fat simplices ($\ast\ast$) & 3 & $O^*(n^{10/7})$  & $\Omega(n^{4/3})$ \\
 &                & $d\ge 4$ & $O(n^{\frac{3}{2}-\frac{1}{O(d^2)}})$ & $\Omega(n^{4/3})$ \\
  \hline
 2 &  balls ($\dag$)  & 3 & $O^*(n^{4/3})$ & \\
  &                & 4 & $O^*(n^{7/5})$ & \\
  &                & 5 & $O^*(n^{7/5})$ & $\Omega(n^{4/3})$ \\
  &                & $d\ge 6$ & $O^*(n^{\frac{3}{2}-\frac{1}{2(2\down{d/2}+1)}})$ & $\Omega(n^{4/3})$\\
  \hline
 2 & fat axis-aligned boxes ($\ddag$)       & $d\ge 3$ & $O(n\log^{d+1}n)$ & $\Omega(n(\frac{\log n}{\log\log n})^{\down{d/2}-1}$)\\
  \hline
3 & general objects ($\ast$) & $d\ge 3$ & $O^*(n^{\frac{3}{2}-\frac{1}{2(2D-1)}})$ & $\Omega(n^{4/3})$\\
\hline
3  & simplices       & 3  & $O^*(n^{10/7})$ & $\Omega(n^{4/3})$ \\
  &                               & $d\ge 4$ & $O(n^{\frac{3}{2}-\frac{1}{O(d^2)}})$ & $\Omega(n^{4/3})$ \\
  \hline
 3 & axis-aligned boxes            & $d\ge 3$ & $O(n\log^{d-1}n)$ & $\Omega(n(\frac{\log n}{\log\log n})^{d-2}$)\\\hline                                                                         
\end{tabular}
\caption{Summary of new results.  Objects for ($\ast$) are assumed to be semi-algebraic for some parameter $D$ associated
with their description complexity.
The upper bounds for ($\ast\ast$) hold more generally for fat polyhedra of constant complexity, while
the lower bounds for ($\ast\ast$) hold more specifically for unit regular tetrahedra or unit non-axis-aligned cubes in ${\mathbb R}^3$.
The lower bounds for ($\dag$) hold more specifically for unit balls.
The lower bound for ($\ddag$) holds more specifically for unit axis-aligned hypercubes.}
\label{tab:new-results}
\end{table}
}

\section{Upper Bounds for 3 Hops}\label{sec:3hop}

In this section, we prove new upper bounds for 3-hop spanners for geometric intersection graphs.
We begin with an easy lemma about general bipartite graphs in a lop-sided case:

\begin{lemma}\label{lem:easy:3hop}
Let $G=(U\cup V,E)$ be a graph where $|U|=m$, $|V|=n$, and $V$ is an independent set.
Then $G$ has a 3-hop spanner of size $O(m^2 + n)$.
\end{lemma}
\begin{proof}
We construct a subgraph $H$ of $G$ as follows. For each $v\in V$ with $\deg(v)\geq 1$, let $e_v$ be an arbitrary edge incident to $v$ in $G$, and add $e_v$ to $H$. Add $G[U]$ to $H$.
Furthermore, for each pair of vertices $u,u'\in U$, let $\pi_{u,u'}$ be a 2-hop path between $u$ and $u'$ in $G$, if such a path exists, and add $\pi_{u,u'}$ to $H$. Clearly, $H$ has $O(m^2+n)$ edges. 
We show that $H$ is a 3-spanner of $G$. By construction, for each edge $e=uu'$ in $G[U]$, the distance between $u$ and $u'$ in $H$ is one. 
Consider an edge $uv$ in $G$ with $u \in U$ and $v \in V$. If $e_v=uv$, then we can go from $u$ to $v$ in 1 hop.  Otherwise, $e_v=u'v$, for some $u'\neq u$. Notice that in that case a $2$-hop path $\pi_{u,u'}$ indeed exists and we can go from $u$ to $v$ in 3 hops by concatenating $\pi_{u,u'}$ with $u'v$.
\end{proof}

The above lemma allows us to almost immediately recover the standard $O(n^{3/2})$ upper bound on 3-hop spanners for general graphs (the greedy algorithm yields the same bound~\cite{althofer1993sparse}):

\begin{corollary}\label{cor:easy:3hop}
Every graph $G=(V,E)$ with $n$ vertices has a 3-hop spanner of size $O(n^{3/2})$.
\end{corollary}
\begin{proof}
Partition the vertices $V=\bigcup_{i=1}^k V_i$ into $k= \sqrt{n}$ groups of size $\sqrt{n}$ each, and apply \Cref{lem:easy:3hop} to each of the graphs $G_i=(V,E_i)$ where $E_i = E(G[V_i]) \cup E[V_i,V\setminus V_i]$. Namely, the edges in each graph comprise the inter-group edges and the bipartite edges between the group and its complement. Take the union of the spanners, with combined size $O(\sqrt{n}\cdot ((\sqrt{n})^2 + n))=O(n^{3/2})$. It is easy to check that the resulting graph is indeed a 3-hop spanner as any original edge belongs to at least one of the graphs $G_i$ as it is either induced by the $U$ part of one of the graphs or a bipartite edge in one of the graphs. 
\end{proof}

\subsection{Semi-algebraic graphs}
Here we improve the general bound $O(n^{3/2})$ for 3-hop spanners for semi-algebraic graphs.
A natural approach is to apply known geometric divide-and-conquer and partitioning techniques to obtain biclique covers of sub-quadratic size, which then yields sub-quadratic upper bounds for 3-hop spanners.
Unfortunately, for many higher-dimensional families of objects, the biclique cover bound can be much worse than $O(n^{3/2})$.
Nevertheless, we show that these divide-and-conquer techniques can work synergistically
with \Cref{lem:easy:3hop} to beat $O(n^{3/2})$: Specifically, we stop the recursion early when subproblems become sufficiently lop-sided and it is advantageous to switch to the $O(m^2+n)$ bound.
In order to apply a divide-and-conquer we apply the polynomial partitioning technique  as follows. 
\begin{theorem}
(Guth and Katz~\cite[Theorem~4.1]{GuthKatz2015})
\label{thm:GK15}
Let $P$ be a set of $m$ points in $\mathbb{R}^D$ for a constant $D\in \mathbb{N}$. For every parameter $1<k<m$, there exists a polynomial $f\in \mathbb{R}[x_1,\ldots , x_D]$ of degree $\deg(f)\in O(k^{1/D})$
such that every connected component of $\mathbb{R}^d\setminus Z(f)$ contains at most $m/k$ points.  
\end{theorem}
This result generalizes to points lying in an irreducible variety. 
\begin{theorem} 
(Fox et al.~\cite[Theorem~4.3]{FPSSZ17})
\label{thm:Fox17}
Let $P$ be a set of $m$ points in a $d$-dimensional irreducible variety $V\subseteq \mathbb{R}^D$ for constant $1\leq d\leq D\in \mathbb{N}$. For every parameter $1<k<m$, there exists a polynomial $f\in \mathbb{R}[x_1,\ldots , x_D]$ of degree $\deg(f)\in O(k^{1/d})$ such that $f$ is not in the (prime) ideal $I(V)$ and every connected component of $\mathbb{R}^d\setminus Z(f)$ contains at most $m/k$ points.  
\end{theorem}
\begin{corollary}\label{cor:polypartition}
    Let $P$ be a set of $m$ points in a $D$-dimensional variety $V\subset \mathbb{R}^D$ for constants $1\leq d\leq D\in \mathbb{N}$. For every parameter $1<k<m$, there exists a partition $P=\bigcup_{i=1}^q P_i$ into $q=O(k)$ sets such that $|P_i|\leq m/k$; and for every $i=1,\ldots q$, there is a connected set $\Delta_i\subset \mathbb{R}^d$, called a \emph{cell}, such that $P_i\subset \Delta_i$ and the zero set $Z(h)$ of any polynomial of bounded degree crosses\footnote{A set $S$ \emph{crosses} a cell $\Delta$ if both $\Delta\cap S$ and $\Delta\setminus S$ are nonempty.} $O(k^{1-1/D})$ cells (consequently, any semi-algebraic set crosses $O(k^{1-1/D})$ cells).
\end{corollary}
\begin{proof}
For constant $D\in \mathbb{N}$, we proceed by induction on $d$. In the base case, $d=0$, is obvious. 
In the indiction step, assume that $d>1$. \Cref{thm:Fox17} yields a polynomial $f\in \mathbb{R}[x_1,\ldots , x_D]$ of degree $\deg(f)\in O(k)$ such that every connected component of $\mathbb{R}^d\setminus Z(f)$ contains at most $m/k$ points. Let every connected component of $\mathbb{R}^d\setminus Z(f)$ be a cell $\Delta_i$ with $P_i=P\cap \Delta_i$. By the Milnor–Thom theorem~\cite{BPR06,Mil64,Tho65}, the number of such cells is $O(k)$, and the zero set of every constant-degree polynomial crosses $O(k^{1-1/D})$ cells~\cite[Theorem~A.2]{ST12}. 
It remains to partition the points in $P\cap Z(f)$. Since $V$ is irreducible and $f$ is not in the ideal of $V$, then $Z(f)\cap V$ is $d'$-dimensional for some $d'<d$. Furthermore, $Z(f)\cap V$ can be written as a union of $O(1)$ irreducible varieties $Z(f)\cap V=V_1\cup \ldots \cup V_t$, each of degree at most $\deg(f)$ \cite{FPSSZ17,GV95}. Associate each variety $V_j$ with the point set $P(j)=P\cap  (V_j\setminus \bigcup_{j'<j} V_{j'})$. By induction, each set $P(j)$ admits a partition into subsets $P_i$ of size at most $m/k$ and corresponding cells $P_i\subset \Delta_i\subset V_j$. Overall, we have partitioned $P$ into $O(k)$ sets, where the constant of proportionality depends on $d$ and $D$.
\end{proof}
\begin{remark}\label{rem:algorithmic}\rm
\Cref{cor:polypartition} uses a decomposition of the zero set of a partitioning polynomial into a union of irreducible varieties. For the algorithmic problem of constructing a partition $P=\bigcup_{i=1}^q P_i$ and corresponding cells $\Delta_i$, one can avoid computing irreducible components with a multilevel partition; see~\cite{MatousekP15} or \cite{AAEKS25}.
\end{remark}

\begin{theorem}\label{thm:3hop}
Let $U$ and $V$ be two sets of $n$ elements.
We are given points $p_1(u),\ldots,p_\ell(u)\in {\mathbb R}^D$ for each $u\in U$ and semi-algebraic sets $S_1(v),\ldots,S_\ell(v)\subset {\mathbb R}^D$ of constant
description complexity for each $v\in V$, where $D$ and $\ell$ are constants.
For any $\ell$-variate Boolean formula $\Phi$, the graph 
 \[
 G_\Phi[U,V]\ :=\ (U\cup V,\ \{(u,v)\in U\times V:\ \Phi([p_1(u)\in S_1(v)],\ldots, [p_\ell(u)\in S_\ell(v)])\})
 \]
has a 3-hop spanner with $O^*(n^{\frac{3D-2}{2D-1}})=O^*(n^{\frac{3}{2}-\frac{1}{2(2D-1)}})$ edges.
\end{theorem}
\begin{proof}
We consider the more general setting, where $|U|=m$ and $|V|=n$.  Let $T_\ell(m,n)$ denote the worst-case optimal size of a 3-hop spanner.

Let $r$ be a parameter.
By \Cref{cor:polypartition}
\footnote{For $D\le 4$, we may alternatively use traditional \emph{$(1/r)$-cuttings}~\cite{AgarwalM94,Chazelle04,Koltun04} (in fact, later in \Cref{thm:3hop:halfspace},
we will switch back to using cuttings), but for $D\ge 5$, optimal bounds on \emph{vertical decompositions} are open and polynomial partitioning gives better results.
},
we can partition ${\mathbb R}^D$ into $O(r^D)$ cells such that
each cell contains at most $m/r^D$ points of $\{p_\ell(u): u\in U\}$,
and each semi-algebraic set $S_\ell(v)$ crosses $O(r^{D-1})$ cells.
For each cell $\Delta$, let $U_\Delta = \{u\in U: p_\ell(u)\in\Delta\}$,
$V_\Delta = \{v\in V: \partial S_\ell(v)\ \mbox{crosses}\ \Delta\}$, 
$V_\Delta' = \{v\in V: S_\ell(v)\ \mbox{contains}\ \Delta\}$, and 
$V_\Delta'' = \{v\in V: S_\ell(v)\ \mbox{does not intersect}\ \Delta\}$.
Arbitrarily partition $V_\Delta$ into $\left\lceil\frac{|V_\Delta|}{n/r}\right\rceil$ subsets $V_\Delta^{(j)}$ of size at most $n/r$ each, and
recursively construct a 3-hop spanner for $G_\Phi[U_\Delta,V_\Delta^{(i)}]$ for each $V_\Delta^{(i)}$. Since $\sum_\Delta |V_\Delta| = O(nr^{D-1})$,
the number of such recursive calls is $O(r^D)$.
Furthermore, for each cell $\Delta$, recursively construct a 3-hop spanner for $G_{\Phi'}[U_\Delta,V_\Delta']$ and a 3-hop spanner for $G_{\Phi''}[U_\Delta,V_\Delta'']$, where $\Phi'$ (resp., $\Phi''$) is the $(\ell-1)$-variate Boolean formula
obtained by setting the $\ell$-th variable to true (resp., false).
The union of these spanners yields a 3-hop spanner for $G_\Phi[U,V]$.
We thus obtain the following recurrence 
for $\ell\geq 1$: 
\begin{equation}\label{eq:recurrence}
    T_\ell(m,n) \ \le\ O(r^D) \left(T_\ell\left(\frac{m}{r^D}, \frac nr\right) + T_{\ell-1}(m,n)\right).
\end{equation}
For base cases, we have $T_0(m,n) = O(m+n)$ (since a biclique $U\times V$ has a 3-hop spanner of linear size, consisting of two stars, centered at a fixed vertex of $U$ and a fixed vertex of $V$), and we can use the upper bound $T_\ell(m,n) \leq O(m^2+n)$ for every $\ell$ by \Cref{lem:easy:3hop}.

Assume inductively that $T_{\ell-1}(m,n)=O^*(m^{\frac{2D-2}{2D-1}}n^{\frac{D}{2D-1}})$ for
$\sqrt{n}\le m\le n^D$.  Choose $r$ to be a large constant.
Expand the recurrence for $k$ levels so that $(m/r^{Dk})^2 \approx n/r^k$, i.e.,
$r^k\approx (m^2/n)^{\frac1{2D-1}}$.  Then the recurrence~\eqref{eq:recurrence} yields
\begin{align*}
T_\ell(m,n)\ &\le\ \sum_{i=1}^k  O(1)^i O(r^D) O(r^{Di}) \left( \frac{m}{r^{Di}}\right)^{\frac{2D-2}{2D-1}}
\left(\frac{n}{r^i}\right)^{\frac{D}{2D-1}} \ +\ O(1)^k  O(r^{Dk}) \left( \left( \frac{m}{r^{Dk}}\right)^2 +
\frac{n}{r^k} \right)\\
&\le\ 
O(1)^kr^D\cdot \big(m^{\frac{2D-2}{2D-1}}n^{\frac{D}{2D-1}} \ +\  nr^{k(D-1)}\big)\\
&=\  O(1)^kr^D\cdot \big(m^{\frac{2D-2}{2D-1}}n^{\frac{D}{2D-1}} \ +\ n (m^2/n)^{\frac{D-1}{2D-1}}\big)
\ =\ O(1)^{\log_r m}r^D \cdot m^{\frac{2D-2}{2D-1}}n^{\frac{D}{2D-1}}.
\end{align*}
Making $r$ an arbitrarily large constant, this yields $T_{\ell}(m,n)=O^*(m^{\frac{2D-2}{2D-1}}n^{\frac{D}{2D-1}})$ for $\sqrt{n}\le m\le n^D$.
\end{proof}

\Cref{thm:3hop} immediately implies an $o(n^{3/2})$ bounds on 3-hop spanners for intersection graphs of any family of objects with constant description complexity.  As an example, we obtain the following corollary.

\begin{corollary}\label{cor:3hop:simplices}
The intersection graph of every set of $n$ simplices (or polyhedra each of constant complexity) in a constant dimension $d$ has a 3-hop spanner with $O\left(n^{\frac{3}{2} - \frac{1}{O(d^2)}}\right)$ edges.
Specifically, for $d=3$, the bound is $O^*(n^{10/7})$.
\end{corollary}
\begin{proof}
Polyhedra of constant complexity can be decomposed into $O(1)$ simplices.
A simplex in $\R^d$ has $d+1$ vertices, so it can be encoded as a point in ${\mathbb R}^{d(d+1)}$, and the set of all simplices intersecting a given simplex maps to a semi-algebraic set in ${\mathbb R}^{d(d+1)}$. We can thus apply \Cref{thm:3hop} with $D=d^2+d$.

We can decrease $D$ with more care.
For example, for $d=3$, two simplices $s$ and $s'$ intersect if and only if 
(a)~a vertex of $s$ is inside $s'$, or
(b)~an edge $pq$ of $s$ intersects a facet $\tau$ of $s'$, or vice versa.
The toughest case concerns (b), which is equivalent to:
(i)~the line through $pq$ intersects $\tau$, and
(ii)~$p$ lies above the plane through $\tau$, and
(iii)~$q$ lies below the plane through $\tau$, or vice versa.
In (i), the line through $pq$ can be encoded as a point in $\R^4$,
whereas in (ii) or (iii), $p$ or $q$ is already a point in $\R^3$.
We can thus apply \Cref{thm:3hop} with $D=4$ (rather than $d(d+1)=12$). 
\end{proof}

\begin{remark}\rm
Interestingly, the above approach implies some new results on Zarankiewicz's problem: upper-bounding the size of a semialgebraic graph, under the assumption that the graph avoids $K_{c,c'}$ as a subgraph for some constants $c$ and $c'$. For example, if the semialgebraic graph $G_\Phi[U,V]$ in \Cref{thm:3hop} avoids $K_{2,c'}$, then it can have at most $O^*(n^{\frac32 - \frac1{2(2D-1)}})$ edges.  This follows essentially from the same proof: In the
lop-sided case with $|U|=m$ and $|V|=n$ (but $V$ is not necessarily an independent set), $O(m^2+n)$ is an easy upper bound on the number of edges of $K_{2,c'}$-free graphs, which we can use in place of \Cref{lem:easy:3hop}.  
More generally, if $G_\Phi[U,V]$ avoids $K_{c,c'}$ with $c\le c'$,
then we can extend the above proof to obtain an upper bound of $O^*(n^{2-\frac1c - \frac{(c-1)^2}{c(cD-1)}})$ on the number of 
edges (using an $O(m^c+n)$ upper bound in the lop-sided case).  
Numerous papers have recently studied Zarankiewicz's problem for geometric graphs, 
but mainly focused on results where the exponents are independent of $c$ (see, e.g., the recent survey paper \cite{smoro2024survey} and the references therein); these upper bounds
are effective only when $c$ is sufficiently large relative to $D$, whereas our bound improves over the general graph bound 
of $O(n^{2-\frac1c})$ for all constant $c\ge 2$.
\end{remark}

\begin{remark}\rm
It is tempting to extend the above approach to larger numbers of hops.  For $2k-1$ hops,
\Cref{lem:easy:3hop} can be replaced by the bound $O(m^{\frac{\up{(k+1)/2}}k}n^{\frac{\down{(k+1)/2}}k} + n+m)$
(by taking a maximal/greedy $(2k-1)$-hop spanner and applying a known combinatorial bound on the size of graphs with girth at least $2k$~\cite{extremal_survey}),
but the above recursion would not provide any improvement over this bound for $k>2$, unfortunately.
\end{remark}

\subsection{Improvement for point-halfspace incidence graphs}

\newcommand{\SSS}{\widehat{S}}

For incidence graphs between points and halfspaces, we can improve the exponent by switching to shallow cuttings~\cite{Mat92} for the geometric divide-and-conquer:

\begin{theorem}\label{thm:3hop:halfspace}
Let $P$ be a set of $n$ points and $S$ be a set of $n$ halfspaces in a constant dimension $D$.
Then their incidence graph $G[P,S] := (P\cup S,\ \{(p,s)\in P\times S: p\in s\})$ has a 3-hop spanner with 
$O^*\left(n^{\frac{3\down{D/2}-2}{2\down{D/2}-1}}\right)
=O^*\left(n^{\frac{3}{2}-\frac{1}{2(2\down{D/2}-1)}}\right)$ edges.
\end{theorem}



\begin{proof}
We will assume that all halfspaces of $S$ are upper halfspaces (since we can handle the lower halfspaces separately and take 
the union of the two 3-hop spanners).
We will consider the more general setting when $|P|=m$ and $|S|=n$.  Let $T(m,n)$ denote the worst-case optimal 3-hop spanner size.

Let $r$ and $t$ be parameters.
We first construct a \emph{$(1/t)$-net} $N\subset P$ of size $O(t\log t)$, i.e., a subset of $P$ such that every \emph{$(m/t)$-deep} halfspace (i.e., a halfspace containing at least $m/t$ points of $P$) contains at least one point of $N$~\cite{HausslerW87}. We construct a 3-hop spanner recursively. For each $(m/t)$-deep halfspace $s\in S$, pick an arbitrary point $p_s\in N$ contained in $s$ and add the edge $p_s s$ to the spanner.
Furthermore, for each pair of vertices $(p,p')\in P\times N$, pick an arbitrary 2-hop path $\pi_{p,p'}$ from $p$ to $p'$ in $G[P,S]$, if such a path exists,
and add $\pi_{p,p'}$ to the spanner. So far, the number of edges in the spanner is $O(n+ mt\log t)$. For every containment pair $(p,s)\in P\times S$, where $s$ is $(n/t)$-deep, we can go from $p$ to $s$ in 3 hops by concatenating $\pi_{p,p_s}$ with $p_s s$. 

Thus, it remains to consider the subset $\SSS$ of all \emph{$(m/t)$-shallow} halfspaces (i.e., halfspaces containing less than $m/t$ points of $P$). 
By the \emph{shallow cutting lemma} of Matou\v sek~\cite{Mat92},
we can form $O(r^{\down{D/2}})$ simplicial cells such that
each cell is crossed by at most $|\SSS|/r$ of the bounding hyperplanes of $\SSS$,
and the cells cover the \emph{$(\le |\SSS|/r)$-level} (i.e., all points $q$ that have at most $|\SSS|/r$ hyperplanes of $\SSS$ below $q$).
Further subdivide the cells so that each cell contains at most $m/r^{\down{D/2}}$ points of $P$; the number of cells remains $O(r^{\down{D/2}})$.
For each cell $\Delta$, let
$P_\Delta = P\cap \Delta$,  
$\SSS_\Delta = \{s\in \SSS: \partial s\ \mbox{crosses}\ \Delta\}$, and
$\SSS'_\Delta = \{s\in \SSS: s\ \mbox{contains}\ \Delta\}$.
Recursively construct a 3-hop spanner for $G[P_\Delta,\SSS_\Delta]$.
Also construct a 3-hop spanner for the biclique $P_\Delta\times \SSS'_\Delta$, of size $O(|P_\Delta|+|\SSS'_\Delta|)$ (consisting of two stars).

We still have to deal with the subset $Q\subset P$ of the remaining points that are not covered by the $(\le |\SSS|/r)$-level.\footnote{There is an alternative, simpler way to deal with these deep points in $Q$: Use a $(1/t)$-net for the halfspaces instead of the
points.  However, this approach will not work for the extension to \Cref{thm:2hop:halfspace} later.
}
Since there are at most $|\SSS|m/t$ containment pairs $(p,s)\in P\times\SSS$ by the shallowness of the halfspaces in $\SSS$, the
size of $Q$ is at most $\frac{|\SSS|m/t}{|\SSS|/r} = rm/t \le m/r^{\down{D/2}}$ by choosing $t=r^{\down{D/2}+1}$.  
We can just arbitrarily partition $\SSS$ into $O(r)$ subsets $\SSS^{(i)}$ of size at most $|\SSS|/r$ and recursively
construct a 3-hop spanner for $G[Q,\SSS^{(i)}]$ for each subset $\SSS^{(i)}$.
We take the union of all the spanners found.

We thus obtain the following recurrence:
\[ 
    T(m,n) \ \le\ O\left(r^{\down{D/2}}\right) T\left(\frac{m}{r^{\down{D/2}}}, \frac nr\right) + O\left(r^{O(1)}(m+n)\right).
\]
Combined with $T(m,n)=O(m^2+n)$ by \Cref{lem:easy:3hop}, the recurrence similarly solves to $T(m,n)=O^*(m^{\frac{2\down{D/2}-2}{2\down{D/2}-1}}n^{\frac{\down{D/2}}{2\down{D/2}-1}})$ for $\sqrt{n}\le m\le n^{\down{D/2}}$.
\end{proof}

As one application, we obtain improved bounds for the case of balls, by a standard lifting mapping from balls to halfspaces:

\begin{corollary}\label{cor:3hop:ball}
Every bipartite intersection graph between $n$ red balls and $n$ blue balls in a constant dimension $d$ has 
a 3-hop spanner with $O^*(n^{\frac{3\down{d/2}+1}{2\down{d/2}+1}})=O^*(n^{\frac{3}{2}-\frac{1}{2(2\down{d/2}+1)}})$ edges.
For example, for $d=3$, the bound is $O^*(n^{4/3})$; for $d\in\{4,5\}$, it is $O^*(n^{7/5})$.
\end{corollary}
\begin{proof}
A red ball with center $(x_1,\ldots,x_d)$ and radius $y$ intersects a blue ball with center $(a_1,\ldots,a_d)$ and radius $b$ 
iff $(x_1-a_1)^2+\cdots+ (x_d-a_d)^2\le (y+b)^2$, 
i.e., the point $(x_1,\ldots,x_d,y,x_1^2+\cdots+x_d^2-y^2)\in\R^{d+2}$ lies inside the
halfspace $\{(x_1,\ldots,x_d,y,z)\in\R^{d+2}: z-2a_1x_1-\cdots-2a_dx_d-2by+a_1^2+\cdots+a_d^2-b^2\le 0\}$.
Thus, we can apply \Cref{thm:3hop:halfspace} with $D=d+2$.
\end{proof}

Note that in the non-bipartite case, an $O(n\log n)$ bound is known for 3-hop spanners~\cite{ChanH23}, by exploiting the fatness of balls.
In the next section, we show that the bound in \Cref{cor:3hop:ball}  actually holds for \emph{2-hop} spanners for balls in the non-bipartite setting.

\section{Upper Bounds for 2 Hops}
\label{sec:2-hop}
As mentioned in the introduction, for many families of objects, e.g., rectangles in $\R^2$, sub-quadratic bounds for 2-hop spanners of their (non-bipartite) intersection graphs are not possible.
In this section, we show that nontrivial upper bounds for 2-hop spanners are possible for intersection graphs of \emph{fat} objects (including balls).

\subsection{When one side forms a clique}

We begin by observing that most of the results in \Cref{sec:3hop} works for 2 hops if one side
of the bipartite graph forms a clique.  This very simple observation will be the key to deriving our fat-object results later:

\begin{lemma}\label{lem:easy:2hop}
Let $G=(U\cup V,E)$ be a bipartite graph where $|U|=m$ and $|V|=n$. Let $G^+$ be 
the union of $G$ and a clique on $U$. Then $G^+$ has a 2-hop spanner of size $O(m^2 + n)$.
\end{lemma}
\begin{proof}
We construct a spanner subgraph $H^+$ of $G^+$ as follows. 
For each vertex $v\in V$, $\deg(v)\geq 1$, let $e_v$ be an arbitrary edge incident to $v$, and add $e_v$ to $H^+$. Furthermore, add the clique on $U$ to $H^+$.
We show that $H^+$ is a 2-hop spanner of $G^+$. Consider an edge $uv$ in $G$.  If $e_v=uv$, then we can go from $u$ to $v$ in 1 hop. Otherwise, say $e_v$ is $u'v$.  
We can go from $u$ to $v$ in 2 hops via the edges $uu'$ and $u'v$.
\end{proof}

The above lemma implies an $O(n^{3/2})$ upper bound on 2-hop spanners in this setting,
by the same grouping trick from the proof of \Cref{cor:easy:3hop}:

\begin{corollary}\label{cor:easy:2hop}
Let $G=(U\cup V,E)$ be a bipartite graph where $|U|+|V|=n$. Let $G^+$ be the union of $G$ and a clique on $U$. Then $G^+$ has a 2-hop spanner of size $O(n^{3/2})$.
\end{corollary}

We also immediately obtain the following analogs to \Cref{thm:3hop,thm:3hop:halfspace}:

\begin{theorem}\label{thm:2hop}
Let $U$ and $V$ be two sets of $n$ elements.
We are given points $p_1(u),\ldots,p_\ell(u)\in \R^D$ for each $u\in U$ and semi-algebraic sets $S_1(v),\ldots,S_\ell(v)\subset \R^D$ of constant
description complexity for each $v\in V$, where $D$ and $\ell$ are constants.
For any $\ell$-variate Boolean formula $\Phi$,
 the graph 
\[G\ =\ (U\cup V,\ \{(u,v)\in U\times V:\ \Phi([p_1(u)\in S_1(v)],\ldots, [p_\ell(u)\in S_\ell(v)])\} \:\cup\: (U\times U))
\]
has a 2-hop spanner with $O^*(n^{\frac{3D-2}{2D-1}})=O^*(n^{\frac{3}{2}-\frac{1}{2(2D-1)}})$ edges.
\end{theorem}
\begin{proof}
We follow the proof of \Cref{thm:3hop}, but using \Cref{lem:easy:2hop} in place of \Cref{lem:easy:3hop}.
We also add a 2-hop spanner for $U\times U$ of $O(m)$ size (namely, a star).
\end{proof}

\begin{theorem}\label{thm:2hop:halfspace}
Let $P$ be a set of $n$ points and $S$ be a set of $n$ halfspaces in a constant dimension~$D$.
Then the graph $G = (P\cup S,\ \{(p,s)\in P\times S: p\in s\} \:\cup\: (P\times P))$ has a 2-hop spanner with 
$O^*(n^{\frac{3\down{D/2}-2}{2\down{D/2}-1}})=O^*(n^{\frac{3}{2}-\frac{1}{2(2\down{D/2}-1)}})$ edges.
\end{theorem}
\begin{proof}
We follow the proof of \Cref{thm:3hop}, but using \Cref{lem:easy:2hop} in place of \Cref{lem:easy:3hop}, and just replacing the path $\pi_{p,p'}$ with an edge $pp'$.
\end{proof}

\subsection{Intersection graphs of fat objects}
\label{ssec:fat}

To obtain our results for fat objects, our idea is to adapt
one of the methods by Chan and Huang~\cite[Section~3.1]{ChanH23}, originally developed for constructing 3-hop spanners of $O(n\log n)$ size. In adapting their method to construct 2-hop spanners, we realize that the main subproblems arising from the recursion correspond to the case when one side forms a clique---a case which fortunately has already been addressed by the
previous subsection! If the fat objects are all of similar size, then a simple grid approach will easily reduce to the one-sided clique case. For fat objects of arbitrary sizes, Chan and Huang~\cite{ChanH23} used a divide-and-conquer approach based on shifted quadtrees.


We use the following definition of fatness~\cite{Chan03}: A family of objects is \emph{fat} if 
for every axis-aligned hypercube $\gamma$ with side length $\ell$, there exist $O(1)$ points hitting all objects that intersect $\gamma$ and have diameter at least $\ell$. There exists a number of different definitions of fatness in the
geometry literature (e.g., see~\cite{BergSVK02, EfratKNS00}). We find the above definition most
appropriate for our purpose. It is not difficult to see that arrangements of Euclidean balls or hyperrectangles of bounded aspect ratios in $\mathbb{R}^d$ are fat according to this definition.

\begin{theorem}\label{thm:2hop:fat:general}
Let $V$ be a set of $n$ fat objects in $\R^d$ for constant~$d$.
Then the intersection graph $G$ of $V$
has a 2-hop spanner with $O(n^{3/2})$ edges.
\end{theorem}
\begin{proof}
Recall that a \emph{quadtree cell} is an axis-aligned box of the form $[i_1/2^k,(i_1+1)/2^k)\times\cdots\times
[i_d/2^k,(i_d+1)/2^k)$ for some integers $i_1,\ldots,i_d,k$.
An object $u$ of diameter $\ell$ is \emph{$C$-aligned} if it is contained in a quadtree cell of side
length $C\ell$.
As argued in~\cite[Section~3.1]{ChanH23}, 
as a consequence of a known shifting lemma~\cite{Chan03}, one can find $O(d)$ vectors $\tau_1,\ldots,\tau_{O(d)}$ with the property that for every pair of objects $u$ and $v$, there exists at least one vector $\tau_j$ such that $u+\tau_j$ and $v+\tau_j$ are both $C$-aligned for some $C=O(d)$.  For each $\tau_j$, it suffices to construct a 2-hop spanner for the subset of all objects $u$ such that $u+\tau_j$ is $C$-aligned; we can then output the union of these $O(d)$ spanners.  From now on, we may thus assume that all objects are $C$-aligned.
 Let $T(n)$ be the worst-case optimal 2-hop spanner size under this assumption.

\begin{enumerate}
\item First find a quadtree cell $\gamma$ such that there are at most $\alpha n$ objects completely inside $\gamma$ and at most $\alpha n$ objects completely outside $\gamma$, with  $\alpha:=\frac{2^d}{2^d+1}$; see~\cite[Lemma~10]{ChanH23} (based on earlier work in~\cite{AryaMNSW98,Chan98}). This cell corresponds to a ``centroid'' of the quadtree.
Recursively construct a 2-hop spanner for the objects completely inside $\gamma$ and a 2-hop spanner for the objects completely outside $\gamma$. 
\item Let $Q_\gamma$ be a set of points hitting all objects that intersect $\partial\gamma$.  By alignedness, these
objects have diameter at least $\ell_\gamma/C$, where $\ell_\gamma$ denotes the side length of $\gamma$.  
By fatness, a hitting set of size $|Q_\gamma|=O(1)$ exists (since $\partial\gamma$ can be covered by $C^d$ hypercubes of side length $\ell_\gamma/C$). 
\item For each point $q\in Q_\gamma$, let $U_q$ be the subset of all objects hit by $q$ (this subset induces a clique).
Let $G_q$ be the bipartite intersection graph between $U_q$ and $V$. 
Construct a 2-hop spanner for $G_q\cup (U_q\times U_q)$ by \Cref{cor:easy:2hop}.
\end{enumerate}

We claim that the union of the spanners found is a 2-hop spanner of $G$.  To see this, consider an edge $uv$ in $G$.
If $u$ and $v$ are both inside $\gamma$ or both outside $\gamma$, then $u$ and $v$ are reachable by 2 hops 
by induction.  Otherwise, one of the objects, say $u$, intersects $\partial\gamma$.  Then $u$ is hit by some point $q\in Q_\gamma$,
and so $uv\in G_q$ and $u$ and $v$ are reachable by 2 hops.

We thus obtain the following recurrence:
\[ T(n)\ \leq\ \max_{n_1,n_2\le \alpha n:\ n_1+n_2\le n} \Big(T(n_1) + T(n_2)\Big) + O(n^{3/2}).
\]
This solves to $T(n)=O(n^{3/2})$.
\end{proof}


We similarly obtain improvements when the fat objects are semialgebraic with constant description complexity:

\begin{theorem}\label{thm:2hop:fat}
Let $V$ be a set of $n$ fat objects in $\R^d$ for constant~$d$.
We are given points $p_1(v),\ldots,p_\ell(v)\in \R^D$ and semi-algebraic set $S_1(v),\ldots,S_\ell(v)\subset \R^D$ of constant description complexity for each $v\in V$, where $D$ and $\ell$ are constants, satisfying the following property:
two objects $u,v\in V$ intersect iff $\Phi([p_1(u)\in S_1(v)],\ldots, [p_\ell(u)\in S_\ell(v)])$, where $\Phi$ is
an $\ell$-variate Boolean formula.
Then the intersection graph $G$ of $V$
has a 2-hop spanner with $O^*(n^{\frac{3D-2}{2D-1}})=O^*(n^{\frac{3}{2}-\frac{1}{2(2D-1)}})$ edges.
\end{theorem}
\begin{proof}
Same as the proof of \Cref{thm:2hop:fat:general}, but using \Cref{thm:2hop} instead of \Cref{cor:easy:2hop}.
\end{proof}

\begin{corollary}
The intersection graph of every set of $n$ fat simplices (or fat polyhedra each of constant complexity, e.g., non-axis-aligned hypercubes) in a constant dimension $d$ has 
a 2-hop spanner with $O(n^{\frac{3}{2} - \frac{1}{O(d^2)}})$ edges.
Specifically, for $d=3$, the bound is $O^*(n^{10/7})$.
\end{corollary}
\begin{proof}
As in the proof of \Cref{cor:3hop:simplices}, this follows by setting $D=d^2$, or in the $d=3$ case, $D=4$.
\end{proof}

\begin{theorem}\label{thm:2hop:fat:halfspace}
Let $V$ be a set of $n$ fat objects in a constant dimension $d$.
We are given a point $p(v)\in \R^D$ and a halfspace $S(v)\subset \R^D$ for each $v\in V$, for a constant $D$, satisfying the following property:
two objects $u$ and $v$ intersect iff $p(u)\in S(v)$.
Then the intersection graph $G$ of $V$ has a 2-hop spanner with $O^*(n^{\frac{3\down{D/2}-2}{2\down{D/2}-1}})=O^*(n^{\frac{3}{2}-\frac{1}{2(2\down{D/2}-1)}})$ edges.
\end{theorem}
\begin{proof}
Same as the proof of \Cref{thm:2hop:fat:general}, but using \Cref{thm:2hop:halfspace} instead of \Cref{cor:easy:2hop}.
\end{proof}

\begin{corollary}\label{cor:2hop:ball}
The intersection graph of $n$ balls in a constant dimension~$d$ has 
a 2-hop spanner with $O^*(n^{\frac{3\down{d/2}+1}{2\down{d/2}+1}})=O^*(n^{\frac{3}{2}-\frac{1}{2(2\down{d/2}+1)}})$ edges.
For example, for $d=3$, the bound is $O^*(n^{4/3})$; for $d\in\{4,5\}$, the bound is $O^*(n^{7/5})$.
\end{corollary}
\begin{proof}
As in the proof of \Cref{cor:3hop:ball}, this follows from \Cref{thm:2hop:fat:halfspace} by setting $D=d+2$.
\end{proof}

For fat axis-aligned boxes, we have the following theorem, which generalizes a result by Conroy and T\'oth~\cite{ConroyT23} from $d=2$ to higher dimensions:

\begin{theorem}\label{thm:2hop:fatbox}
For every $d\geq 2$, the intersection graph of $n$ fat axis-aligned boxes (e.g., axis-aligned hypercubes) in a constant dimension~$d$ has 
a 2-hop spanner with $O(n\log^{d+1}n)$ edges.
\end{theorem}
\begin{proof}
Same as the proof of \Cref{thm:2hop:fat}, but to construct a 2-hop spanner for $G_p\cup (U_q\times U_q)$, 
we use a known \emph{biclique cover} construction for
the bipartite intersection graph $G_p$ of boxes: this gives
a collection of bicliques of the form $U_i\times V_i$, covering all the edges of $G_p$,
with total size $\sum_i (|U_i|+|V_i|) = O(n\log^{d}n)$.
(The construction is obtained by range-tree-style divide-and-conquer; e.g., see \cite{Chan06}.)
When the clique $U_i\times U_i$ is added to the biclique $U_i\times V_i$, there is a trivial
2-hop spanner of size $O(|U_i|+|V_i|)$ (namely, take a star centered at an arbitrary point in $U_i$).
We just take the union of these spanners (together with a linear-size 2-hop spanner for $U_q\times U_q$) 
and  obtain a spanner of size $O(n\log^{d}n)$ for $G_p\cup (U_q\times U_q)$.
The recursion in the proof of \Cref{thm:2hop:fat} causes one more logarithmic factor.
\end{proof}

\section{A 3-Hop Upper Bound for Axis-Aligned Boxes}

In this section, we note that the known $O(n\log n)$-size 3-hop spanner construction for axis-aligned rectangles in $\R^2$ by Chan and Huang~\cite{ChanH23} can be extended to give an $O(n\log^{d-1}n)$-size 3-hop spanner for axis-aligned boxes in $\R^d$.
(One could alternatively use biclique covers, as in the proof of \Cref{thm:2hop:fatbox}, but this would yield a slightly larger
$O(n\log^dn)$ bound.)

\begin{theorem}
The intersection graph of every set of $n$ axis-aligned boxes in a constant dimension $d$ has a 3-hop spanner with $O(n\log^{d-1}n)$ edges.
\end{theorem}
\begin{proof}
We will prove the result more generally for bipartite intersection graphs and use a divide-and-conquer strategy similar to range trees or segment trees.

Let $\sigma$ be a slab orthogonal to the $d$-th axis.
We say that an axis-aligned box $s$ is \emph{short} for $\sigma$ if it has at least one corner point strictly inside $\sigma$,
and $s$ is \emph{long} for $\sigma$ if it intersects $\sigma$ but all its corner points are outside or on
the boundary of $\sigma$.

Suppose we are given two sets $U$ and $V$ of axis-aligned boxes in $\R^d$, all short in $\sigma$,
with a total of $N$ corner points strictly inside $\sigma$.  Note that since the boxes are short, the number of boxes must be $O(N)$.
We want to construct a 3-hop spanner for the bipartite intersection graph $G[U,V]$ between $U$ and $V$.

\newcommand{\llong}{\textrm{long}}
\newcommand{\short}{\textrm{short}}
Divide $\sigma$ into two subslabs $\sigma_1$ and $\sigma_2$, each containing $n/2$ corner points.
For each $j\in\{1,2\}$, let $U_j$ (resp.\ $V_j$) be the boxes of $U$ (resp.\ $V$) intersecting $\sigma_j$, clipped to $\sigma_j$.
Let $U_j^\short$ (resp.\ $V_j^\short$) be the boxes of $U_j$ (resp.\ $V_j$) that are short for $\sigma_j$.
Let $U_j^\llong$ (resp.\ $V_j^\llong$) be the boxes of $U_j$ (resp.\ $V_j$) that are long for $\sigma_j$.

For each $j\in\{1,2\}$:
\begin{enumerate}
\item Recursively construct a 3-hop spanner for $G[U_j^\short,V_j^\short]$.
\item Recursively construct a 3-hop spanner for $G[U_j,V_j^\llong]$ by projection to the first $d-1$ coordinates.
\item Recursively construct a 3-hop spanner for $G[U_j^\llong,V_j]$ by projection to the first $d-1$ coordinates.
\end{enumerate}

We then take the union of the spanners found.  As a result, we then obtain the following recurrence for the spanner size:
\[ T_d(N) \:\le\: 2\,T_d(N/2) + O(1) T_{d-1}(N).
\]
For the base case, we have $T_2(N)=O(N\log N)$ by Chan and Huang~\cite{ChanH23}.
The recurrence solves to $T_d(N)=O(N\log^{d-1}N)$.
\end{proof}

\section{Lower Bounds}

In this section, we prove lower bounds for 3-hop and 2-hop spanners, by relating the problem to the combinatorial question
of bounding the size of geometric graphs avoiding $K_{2,2}$.
Our lower bounds for 3-hop spanners are obtained from the following simple observation:

\begin{observation}\label{obs:lb:3hop}
If a bipartite graph $G$ does not contain any 4-cycle, i.e., $K_{2,2}$, then any 3-hop spanner of $G$ must include all its edges.
\end{observation}

For 2-hop spanners, one could similarly derive lower bounds from constructions of graphs without 3-cycles (i.e., triangle-free graphs).
However, for fat objects, triangle-free intersection graphs actually are known to have linear size.
Instead, we propose the following lemma, showing that constructions of large bipartite intersection graphs between two sets of objects that avoid $K_{2,2}$ (where objects within each set
are allowed to intersect) can also yield lower bounds for 2-hop spanners for (not-necessarily bipartite) intersection graphs:

\newcommand{\GG}{H}

\begin{lemma}\label{lem:lb:2hop}
If $G=(U\cup V,E)$ is a graph (with $U\cap V=\emptyset$), and the bipartite subgraph $E'=E\cap (U\times V)$ does not contain $K_{2,2}$,
then any 2-hop spanner of $G$ must have at least $|E'|$ edges.
\end{lemma}
\begin{proof}
    Basically, we want to show that extra edges in $U\times U$ and $V\times V$ do not help.
    We use a charging argument. Let $\GG$ be a 2-hop spanner of $G$. For each edge $uv\in E'$ with $u\in U$ and $v\in V$, if $uv\in \GG$, we charge the edge $uv$ to itself. Otherwise, $\GG$ must contain edges $ux$ and $xv$ for some $x$.  Without loss of generality, say $x\in U$. We charge the edge $uv$ to the edge $ux$.
    
    Since the total charge is $|E'|$, it suffices to show that no edge in $\GG$ receives more than one charge. Suppose that some edge $e$ in $\GG$ receives charges from two different edges $f_1,f_2\in E'$. For $f_1$ to be charged to $e$, there must exist an edge $g_1\in E'$ such that $e,f_1,g_1$ form a triangle. Similarly, there must exist an edge $g_2\in E'$ such that $e,f_2,g_2$ form a triangle. Since $f_1,f_2\not\in\GG$ and $g_1,g_2\in \GG$, the edges $f_1,g_1,f_2,g_2$ are distinct
    and form a $K_{2,2}$ in $E'$: a contradiction.
\end{proof}

We can use known combinatorial results to construct $K_{2,2}$-free intersection graphs:

\begin{lemma}\label{lem:K22}  
For every positive integer $n$, the following hold:
\begin{enumerate}
\item[(a)] There exist $n$ points and $n$ lines in $\R^2$ whose incidence graph does not contain $K_{2,2}$ and has $\Omega(n^{4/3})$ edges.
\item[(b)] There exist $n$ tetrahedra in $\R^3$ whose intersection graph is bipartite, does not contain $K_{2,2}$, and has $\Omega(n^{4/3})$ edges.
\item[(c)] There exist $n$ red/blue congruent regular tetrahedra (or congruent non-axis-aligned cubes) in $\R^3$ such that the bipartite intersection graph between the red tetrahedra and the blue tetahedra does not contain $K_{2,2}$ and has $\Omega(n^{4/3})$ edges.
\item[(d)] There exist $n$ points and $n$ halfspaces in $\R^5$ whose incidence graph does not contain $K_{2,2}$ and has $\Omega(n^{4/3})$ edges.
\item[(e)] There exist $n$ red/blue congruent balls in $\R^5$ such that the bipartite intersection graph between the red balls and the blue balls does not contain $K_{2,2}$ and has $\Omega(n^{4/3})$ edges.
\item[(f)] There exist $n$ points and $n$ axis-aligned boxes in $\R^d$ whose incidence graph does not contain $K_{2,2}$ and has
$\Omega(n (\log n/\log\log n)^{d-1})$ edges.
\item[(g)] There exist $n$ axis-aligned boxes in $\R^d$ whose intersection graph is bipartite, does not contain $K_{2,2}$, and has $\Omega(n (\log n/\log\log n)^{d-2})$ edges.
\item[(h)] There exist $n$ red/blue congruent axis-aligned hypercubes in $\R^d$ such that the bipartite intersection graph between the red hypercubes and blue hypercubes does not contain $K_{2,2}$ and has $\Omega(n (\log n/\log\log n)^{\down{d/2}-1})$ edges.
\end{enumerate}
\end{lemma}
\begin{proof}\

\begin{enumerate}
\item[(a)] This is well-known, by a construction of Erd\H os~\cite{Erdos85} (see also~\cite{BalkoST23,Elekes}) and incidence graphs between points and lines in $\R^2$ automatically do not contain $K_{2,2}$.
\item[(b)]  Let $P$ be the set of points and $L$ be the set of lines in $\R^2$ from (a).
Map each point $p\in P$ to the red vertical line $\varphi(p) = p\times (-\infty,\infty)$ in $\R^3$.
Map the $i$-th line $\ell\in L$ to a blue line $\psi(\ell) = \ell\times \{i\}$ in $\R^3$ (which lies
on the horizontal plane $z=i$).
Then $\varphi(p)$ intersects $\psi(\ell)$ iff $p$ lies on $\ell$.
Moreover, there are no red-red or blue-blue intersections.
The result then follows by viewing these red/blue lines as thin tetrahedra.
\item[(c)] Let $P$ be the set of points and $L$ be the set of lines in $\R^2$ from (a).
We may clip each line $\ell\in L$ to a line segment of a sufficiently large length $M$.
Map each point $p\in P$ to some red regular tetrahedron $\varphi(p)$ in $\R^3$ 
that has side length $M$, lies inside the halfspace $z\ge 0$, and touches the plane $z=0$ precisely at the point $p\times \{0\}$.
Map each line segment $\ell\in P$ to some blue regular tetrahedron $\psi(\ell)$ in $\R^3$ 
that has side length $M$, lies inside the halfspace $z\le 0$, and touches the plane $z=0$ precisely at the line segment $\ell\times \{0\}$.
Then $\varphi(p)$ intersects $\psi(\ell)$ iff $p$ lies on $\ell$.
The construction is similar for non-axis-aligned cubes.
\item[(d)] This was noted by Chan and Har-Peled~\cite{CH23}.
Specifically, let $P$ be the set of points and $L$ be the set of lines in $\R^2$ from (a).
A point $(x,y)\in P$ is incident on a line in $L$ with equation $ax+by=1$ iff
$(ax+by-1)^2 \le \delta$, i.e., $a^2x^2 + 2abxy+b^2y^2 -2ax-2by + 1\le \delta$, for a sufficiently small $\delta>0$.
Map each point $p=(x,y)\in P$ to a point $\varphi(p)=(x^2,xy,y^2,x,y)\in\R^5$.
Map each line $\ell\in L$ with equation $ax+by=1$ to a halfspace $\psi(\ell)=\{(\xi_1,\xi_2,\xi_3,\xi_4,\xi_5)\in\R^5:
a^2\xi_1 + 2ab\xi_2 + b^2\xi_3 -2a\xi_4 - 2b\xi_5 + 1\le \delta\}$.
Then $\varphi(p)$ lies inside $\psi(\ell)$ iff $p$ lies on $\ell$.
\item[(e)]
Let $P$ be the set of points and $S$ be the set of halfspaces in $\R^5$ from (d).
We can view each halfspace of $S$ as a giant ball, of a fixed and sufficiently large radius $M$.
Now map each point $p\in P$ to a red ball $\varphi(p)$ centered at $p$ of radius $M/2$.
Map each ball $s\in S$ to a blue ball $\psi(s)$ with the same center of radius $M/2$.
Then $\varphi(p)$ intersects $\psi(s)$ iff $p$ lies inside $s$.
\item[(f)]
This follows from a construction by Chazelle~\cite{Chazelle90}, as noted by Chan and Har-Peled~\cite{CH23}.
\item[(g)]
Let $P$ be the set of points and $S$ the set of boxes in $\R^{d-1}$ from (f).
Map each point $p\in P$ to a red vertical line $\varphi(p)= p\times (-\infty,\infty)$ in $\R^d$.
Map the $i$-th box $s\in S$ to a blue box $\psi(s) = s\times \{i\}$ in $\R^d$.
Then $\varphi(p)$ intersects $\psi(s)$ iff $p$ is inside $s$.
Moreover, there are no red-red or blue-blue intersections.
The result then follows by viewing the red lines as thin boxes.
\item[(h)]
Say $d$ is even.
Let $P$ be the set of points and $S$ be the set of boxes in $\R^{d/2}$ from (f).
Map each point $p=(x_1,\ldots,x_{d/2})\in P$ to a red hypercube $\varphi(p)= [-x_1,-x_1+M]\times [x_1,x_1+M] \times\cdots\times [-x_{d/2},-x_{d/2}+M]\times [x_{d/2},x_{d/2}+M]$ in $\R^d$ for a sufficiently large side length $M$.
Map each box $s=[a_1,b_1]\times\cdots\times [a_{d/2},b_{d/2}]\in S$ to a blue hypercube $\psi(s)=[-a_1-M,-a_1]\times [b_1-M,b_1] \times\cdots\times
[a_{d/2}-M,a_{d/2}]\times [b_{d/2}-M,b_{d/2}]$ in $\R^d$.
Then $\varphi(p)$ intersects $\psi(s)$ iff $p$ lies inside $s$.
\end{enumerate}
%
\end{proof}

Lower bounds for 3-hop and 2-hop spanners for a variety of intersection graphs now immediately follow.

\begin{theorem}\label{thm:lb}
For every positive integer $n$, the following holds:
\begin{enumerate}
\item[(i)] There exist $n$ tetrahedra in $\R^3$ such that any 3-hop spanner of their intersection graph requires $\Omega(n^{4/3})$ edges.
\item[(ii)] There exist $n$ congruent regular tetrahedra (or congruent non-axis-aligned cubes) in $\R^3$ such that any 2-hop spanner of their intersection graph requires $\Omega(n^{4/3})$ edges.
\item[(iii)] There exist $n$ congruent balls in $\R^5$ such that any 2-hop spanner of their intersection graph requires $\Omega(n^{4/3})$ edges.
\item[(iv)] There exist $n$ axis-aligned boxes in $\R^d$ such that any 3-hop spanner of their intersection graph requires $\Omega(n (\log n/\log\log n)^{d-2})$ edges.
\item[(v)] There exist $n$ congruent axis-aligned hypercubes in $\R^d$ such that any 2-hop spanner of their intersection graph requires $\Omega(n (\log n/\log\log n)^{\down{d/2}-1})$ edges.
\end{enumerate}
\end{theorem}
\begin{proof}\

\begin{enumerate}
\item[(i)] By \Cref{lem:K22}(b) and \Cref{obs:lb:3hop}.
\item[(ii)] By \Cref{lem:K22}(c) and \Cref{lem:lb:2hop}.
\item[(iii)] By \Cref{lem:K22}(e) and \Cref{lem:lb:2hop}.
\item[(iv)] By \Cref{lem:K22}(g) and \Cref{obs:lb:3hop}.
\item[(v)] By \Cref{lem:K22}(h) and \Cref{lem:lb:2hop}.\qedhere
\end{enumerate}
%
\end{proof}

\begin{remark}\rm
It remains open to prove better lower bounds for simplices in dimensions $d>3$ or balls in dimensions $d>5$, ideally with exponent converging to $3/2$ as $d$ increases.
It is possible to adapt \Cref{obs:lb:3hop} and \Cref{lem:lb:2hop} to avoid $K_{2,c}$ instead of $K_{2,2}$ for any constant~$c$ (while losing some $c^{O(1)}$ factors), but we are not aware of any stronger lower bounds on geometric graphs avoiding $K_{2,c}$ that could be exploited here (except in some lop-sided bipartite cases).

\end{remark}

One common property to all the geometric intersection graphs mentioned in this paper is that they all have bounded VC-dimension\footnote{The \emph{VC-dimension} of a graph is defined as the VC-dimension of the set system induced by the neighborhoods of its vertices.}. As already noted in the introduction, bounded VC-dimension alone does not guarantee a $3$-spanner with $o(n^{3/2})$ edges, as demonstrated in the following theorem:
\begin{theorem}\label{thm:bipartite-boundedVC-LB}
For any $n\in\mathbb{N}$, there exists a graph $G=(V_1 \cup V_2,E)$
with VC-dimension at most $2$ so that any $3$-spanner for $G$ has $\Omega(n^{3/2})$ edges.
\end{theorem}
\begin{proof}
    The construction consists of the known constructions of bipartite graphs without $C_4$ with $\Omega(n^{3/2})$ edges (e.g., the construction with projective planes over finite fields $(F_p)^2$). Its easy to see that the VC-dimension of such graphs is at most 2. By \Cref{obs:lb:3hop}, any $3$-spanner must use all bipartite edges.
\end{proof}

\begin{remark}
Bounded VC-dimension is known to imply the existence of spanning trees with sublinear crossing number~\cite{ChazelleW89},
biclique covers with subquadratic size, and various other properties often used in geometric range searching.  
Along the same lines, one may wonder whether general set systems with bounded VC-dimension admit some analog to $(1/r)$-cuttings~\cite{Chazelle04}
or polynomial partitioning.  Interestingly, our work provides a negative answer to this question: An analog to cuttings would
imply an $o(n^{3/2})$ upper bound for 3-hop spanners, by following the method in
the proof of \Cref{thm:3hop,thm:3hop:halfspace}, and this would contradict \Cref{thm:bipartite-boundedVC-LB}.
\end{remark}

\section{Open Problems}
\label{sec:con}

The following problems remain open.

\begin{enumerate}

    \item Is there a 2-spanner of size $o(n^{4/3})$ for $n$ balls in $\R^3$?
     \Cref{cor:2hop:ball} gives an upper bound of $O^*(n^{4/3})$ in $\R^3$ and \Cref{thm:lb}(iii) gives a lower bound of $\Omega(n^{4/3})$ but in $\R^5$. The current best lower bound in $\R^3$ is only $\Omega(n\frac{\log n}{\log \log n})$ that holds already for $\R^2$~\cite{ConroyT23}.
        
   \item Is there a 2-spanner of size $o(n^{3/2})$ for general fat objects in the plane, not necessarily of constant description complexity? 
    In \Cref{sec:2-hop}, we proved the upper bound $O(n^{3/2})$. The near-linear bound
   $n\log n\cdot 2^{O(\log^*n)}$ in \Cref{tab:old-results} holds only for fat objects of constant algebraic description complexity~\cite{AronovBES14}. 

   \item Is there a 2-spanner of size $O(n^{3/2-\delta})$ for any fat objects in $\R^3$ for an absolute constant $\delta>0$ independent of the description complexity? In \Cref{sec:2-hop}, we proved the upper bound $O(n^{3/2})$, but the current best lower bound is only $\Omega(n^{4/3})$, which holds more specifically for unit regular tetrahedra in $\R^3$.

    \item Close the gap between the lower and upper bounds for 3-spanners of several important families of geometric objects in the plane: The current bounds are 
   (i)  $\Omega(n)$ and $O(n\log^3 n)$ for string graphs; 
   (ii) $\Omega(n)$ and $O(n\log n)$ for arbitrary fat convex bodies;
   (iii) $\Omega(n \log n/ \log \log n)$ and $O(n\log n)$ for axis-aligned squares; and 
   (iv) $\omega(n)$ and $O(n\log n)$ for axis-aligned rectangles; see \Cref{tab:old-results}.

   \item Close the gap between the $\Omega(n^{4/3})$ lower bound and the
   $O^*(n^{10/7})$ upper bound for 3-spanners for tetrahedra in $\R^3$.
\end{enumerate}

\bibliographystyle{alphaurl}
\bibliography{main.bib}

\end{document}